\documentclass[journal,twocolumn]{IEEEtran}
\usepackage{amsfonts}    
\usepackage{color}
\usepackage{graphicx}
\usepackage[dvips]{epsfig}
\usepackage{graphics}
\usepackage{cite}
\usepackage{multirow} 
\usepackage{enumerate}
\usepackage{arydshln}
\usepackage{amsthm}
\usepackage{amsmath} 
\usepackage{amssymb}  

\usepackage{amsfonts}    
\usepackage{color}
\usepackage{graphicx}
\usepackage[dvips]{epsfig}
\usepackage{graphics}
\usepackage{cite}
\usepackage{arydshln}
\usepackage{amsmath} 
\usepackage{amssymb}  
\usepackage{algorithm}
\usepackage{algorithmic}
\usepackage{cite}
\usepackage[square, comma, sort&compress, numbers]{natbib}

\def\rainfty{\rightarrow\infty}

\def\bN{\mathbb{N}}

\def\cQ{\mathbb{Q}}

\def\cE{\mathcal{E}}

\def\cG{\mathcal{G}}

\def\cJ{\mathcal{J}}

\def\cN{\mathcal{N}}

\def\cQ {\mathcal{Q}}

\def\cV{\mathcal{V}}

\def\bee{\begin{equation}}
\def\ene{\end{equation}}
\def\beq{\begin{eqnarray}}
\def\enq{\end{eqnarray}}

\newtheorem{ass}{Assumption}
\newtheorem{lem}{Lemma}
\newtheorem{rem}{Remark}
\newtheorem{cor}{Corollary}
\newtheorem{thm}{Theorem}

\newtheorem{exmp}{Example}
\newtheorem{defi}{Definition}

\begin{document}
\title{How to Stop Consensus Algorithms, {\em locally}?}
\author{Pei Xie, Keyou You and Cheng Wu \thanks{This work was supported by the National Natural Science Foundation of China (61304038),  and Tsinghua University Initiative Scientific Research Program.} \thanks{P. Xie, K. You, and C. Wu are with the Department of Automation and TNList, Tsinghua University, 100084, China (emails: xie-p13@mails.tsinghua.edu.cn, \{youky, wuc\}@tsinghua.edu.cn).}}

\maketitle
\begin{abstract}
This paper studies problems on {\em locally} stopping distributed consensus algorithms over networks where each node updates its state by interacting with its neighbors and decides by itself whether certain level of agreement has been achieved among nodes. Since an individual node is unable to access the states of those beyond its neighbors, this problem becomes challenging.  In this work, we first define the stopping problem for generic distributed algorithms.  Then, a distributed algorithm is explicitly provided for each node to stop consensus updating by exploring the relationship between the so-called {\em local} and {\em global} consensus. Finally, we show both in theory and simulation that its effectiveness depends both on the network size and the structure. 
\end{abstract}

\section{Introduction}
This work proposes local stopping rules for distributed algorithms with a network of computing nodes, and each node updates its states by only interacting with neighbors. Due to the advantages of cooperation among nodes,  they have been widely used in many areas, including distributed localization \cite{sheu2008distributed}, cognitive networks \cite{bazerque2010distributed}, cooperative control of unmanned air vehicles (UVAs) \cite{richards2002coordination}, flocking \cite{dong2016flocking}, and social networks \cite{acemoglu2011opinion}. Since optimality of many distributed algorithms is attained in an asymptotic sense, it is critical for each node of the algorithm to decide {\em by itself} when to stop updating, which to the best of our knowledge, has not be well studied in the literature. This problem is substantially different from that of a centralized algorithm as it lacks a centralized coordinator to monitor the network.

Most of existing stopping rules are based on that each node will stop updating after an adequately large number of iterations. Clearly, this idea is very easy to implement, but also suffers several limitations. Firstly,  the number of iterations depends on quite a few factors, e.g., the network topology, the initial network state, and the convergence rate of distributed algorithms, all of which essentially require {\em global} information of the system. For instance, the convergence rate of the celebrated consensus algorithm is quantified by the algebraic connectivity of the interaction graph and the initial network state, both of which are difficult to evaluate in a distributed way and also require much more computational and communication overheads.  Secondly, the iteration number is always derived from the worst-case point of view, and its conservativeness is case-by-case dependent, which is unable to evaluate. The last but not the least,  the number of iterations is impossible to obtain in advance for  time-varying networks, especially for stochastically time-varying networks due to the causality issue.

In this work, we formally define the local stopping problem of a distributed algorithm over networks under the following constraints: (i) each node is only able to access information from its neighbors. This enables the rule  scalable with the network size. (ii) To save storage capacity, it is not allowed for nodes to store ``large" data, e.g., the historical state information. (iii) the rule can be easily checked by an individual node, and (iv) the communication overhead to support the rule is as small as possible. 

Then, we propose a fully distributed stopping scheme for general distributed algorithms by allowing each individual node in the network to maintain a vector, which is updated iteratively by using  information from its in-neighbors to gradually ``learn'' the network. Each node will stop updating based on this auxiliary vector.  

For distributed consensus algorithms, the earliest concern about this topic has been addressed in  \cite{yadav2007distributed}, which requires the nodes to simultaneously run minimum and maximum protocols along with the consensus protocol, and both the maximum and minimum protocols are re-initialized per $D$-iterations where $D$ is the diameter of the graph. The drawback of the method is  that it requires to communicate three different states of the same dimension at each time slot. This increases the communication burden, especially when the state of a node is a multi-dimensional vector. Moreover, the synchronization problem in re-initializing the minimum and maximum protocols remains difficult in practice. To overcome these issues, an adaptive stopping-recalling mechanism is designed in \cite{manitara2014distributed} to stop the push-sum consensus algorithm in finite time. The method permits the ``inactive'' nodes to recover transmitting again, but the problem is that a node never know whether it will reactivate or not even though the whole network has reached consensus, and has to await orders all the time.

In fact, this paper combines the advantages of the above two methods. We adopt the idea of minimum-consensus \cite{cortes2008distributed} to check the {\em uniformly local consensus} to verify the global consensus. Different from that in \cite{yadav2007distributed}, our method needs not to reset minimum-consensus protocol, and the state of the minimum consensus  is  a discrete-valued scalar even if each node maintains a vector state. To evaluate the performance of our method, we characterize the extra time that the nodes spend to locally detect the global consensus.  
 
An outline of this paper is as follows. In Section \ref{formulation}, we propose the local stopping problem in distributed algorithms. In Section \ref{consensus}, a particular local stopping rule for the consensus algorithm is designed. In Section \ref{simulation}, we analyze the sensitivity of the method. Finally, some concluding remarks are drawn in Section \ref{conclusion}.

\textbf{Notation: }Consider a directed graph (digraph) $\mathcal{G}=(\mathcal{V},\mathcal{E})$, where $\mathcal{V}$ is the set of nodes and $\mathcal{E}$ is the set of directed edges $(i,j)$ with $i,j\in\mathcal{V}$. For node $i\in\mathcal{V}$, its {\em in-neighbors} and {\em out-neighbors} are respectively defined by $\mathcal{N}_i^{-}:=\{j\in\mathcal{V}|(i,j)\in\mathcal{E}\}$ and $\mathcal{N}_i^{+}:=\{j\in\mathcal{V}|(j,i)\in\mathcal{E}\}$. For two nodes $i_1,i_j\in\mathcal{V}$, the {\em directed path} from $i_j$ to $i_1$ is defined by $\{(i_1,i_2),\ldots,(i_{j-1},i_j)\}$ with $(i_k,i_{k+1})\in\mathcal{E}$ for $k\in[1, j-1]$. We denote $d_{ij}$ as the {\em distance} from $j$ to $i$, i.e., the length of the shortest path from $j$ to $i$. The digraph is said to contain a {\em spanning tree} if it has a {\em root} node that is connected to any other node in the graph via a directed path, and is {\em strongly connected} if each node is connected to every other node in the graph via a directed path. The diameter of the strongly connected digraph $\mathcal{G}$ is defined by $D=\max_{i,j\in\mathcal{V}} d_{ij}$. Finally, we denote $x_A\triangleq\{x_i|i\in A\}$.

\section{Local Stopping Scheme for Networked Distributed Algorithms}
\label{formulation}
\subsection{Distributed Algorithms}
Let $x_i(k)$ be the state of node $i$ at time slot $k$ in a fixed graph\footnote{It can be easily extended to time-varying graphs.} $\cG=\{\cV,\cE\}$. We are concerned with distributed algorithms of the following form
\begin{equation}
\label{basic_dis}
x_i(k+1)=h_i(x_{\mathcal{N}_i^-\cup\{i\}}(k), k),
\end{equation}
where $h_i(\cdot, k)$ is a function of $x_i(k)$ and $x_{\mathcal{N}_i^-}(k)$, both of which are directly known to node $i$ at time slot $k$. In different applications, $h_i(\cdot, k)$ may take different forms. Clearly, the update law in (\ref{basic_dis}) is very generic, and contains a lot of important algorithms.  Here we present three interesting applications for an illustration.
\begin{exmp}[Consensus algorithm] \label{consensus}In discrete consensus algorithms \cite{olfati2004consensus}, $h_i(\cdot, k)$ is simply a convex combination of $x_i$ and $x_{\mathcal{N}_i^-}$ , i.e., \begin{equation}
\label{consensus_dis}
x_i(k+1)=\sum_{j\in\mathcal{N}_i^-\cup\{i\}}\alpha_{ij} x_j(k),
\end{equation}
where $\alpha_{ij}>0$ if and only if $j\in\mathcal{N}_i^-\cup\{i\}$, and zero, otherwise. Moreover, $\sum_{j\in\mathcal{N}_i^-\cup\{i\}}\alpha_{ij}=1$. Then, the network state asymptotically converges to consensus (i.e., $\lim_{k\rainfty}\|x_i(k)-x_j(k)\|=0$) if and only if $\cG$ contains a spanning tree \cite{ren2005consensus}.
\end{exmp}
\begin{exmp}[Opinion dynamics] In Friedkin-Johnsen model of opinions evolution \cite{parsegov2016novel}, the state of node $i$ is updated as
\begin{equation}
x_i(k+1)=\omega_i \sum_{j\in\mathcal{N}_i^-\cup\{i\}}\alpha_{ij} x_j(k)+(1-\omega_i)x_i(0),
\end{equation}
where $0\le \omega_i\le 1$ represents the ``coupling condition" and $x_i(0)$ is referred to as the node prejudice. Under certain conditions, the network state asymptotically converges \cite{parsegov2016novel}. 
\end{exmp}

\begin{exmp}[Distributed gradient descend] The aim of the distributed gradient descend (DGD) is to solve the following optimization problem
\bee
\min\sum_{i\in\cV} f_i(x),\label{sp}
\ene
where $f_i(\cdot)$ is a convex function only known by node $i$. In the DGD algorithm, $h_i(\cdot)$  is given as follows
\begin{equation}
x_i(k+1)=\sum_{j\in\mathcal{N}_i^-\cup\{i\}}\alpha_{ij} x_j(k)-\zeta(k)\cdot d_i(k),
\end{equation}
where $\zeta(k)$ is a sequence of appropriately selected stepsizes, $d_i(k)$ is a (sub)-gradient of $f_i$ at $x_i(k)$ and $\alpha_{ij}$ is given as in (\ref{consensus_dis}). By \cite{nedic2009distributed}, it follows that if $\cG$ is balanced and strongly connected, the state of each node asymptotically converges to some common optimal point of the optimization problem (\ref{sp}).
\end{exmp}

In the above important applications, the convergence of the network state is shown in the asymptotic sense.  Then, a natural question is when to stop updating (\ref{basic_dis}) for an individual node  to achieve a desired level of ``optimality"?  If there were a central authority to monitor the network state, the problem is obviously trivial. However, in distributed algorithms over networks, an individual node is unable to directly observe the whole network state, which is quite different from the centralized algorithms, and the local stopping problem requires further investigation. 

A naive idea is that the node can simply stop updating (\ref{basic_dis}) after a sufficiently large number of iterations. Indeed, this is very easy to implement. However, to ensure a desired quality of the node state, the number of iterations usually depends on quite a few factors,  most of which are global information of the system. Take the consensus algorithm in (\ref{consensus_dis})
as an example. Let $d(k)=\max_{i,j\in\cV}\{\|x_i(k)-x_j(k)\|_{\infty}\}$, it follows from \cite{olfati2004consensus} that $d(k)\le c \rho^k d(0)$ where $c>0$ is positive and independent of network. The decaying rate $\rho$ is the second largest eigenvalue for a symmetric stochastic matrix $A=(\alpha_{ij})$, 
and strictly less than one if the undirected graph $\cG$ is connected.  To ensure $d(k)\le \epsilon$, the number of iterations is usually set to satisfy that 
 \bee
 k\ge\big\lceil\frac{\log\epsilon/(cd(0))}{\log \rho}\big\rceil.
 \ene

Clearly, $d(0)$ and $\rho$ in the lower bound respectively depend on the initial network state and the network structure, both of which are unavailable to any individual node.  One may argue that $d(0)$ can be obtained by using distributed maximum- and minimum-consensus algorithms.  However, distributedly estimating $\rho$ is nontrivial \cite{di2014distributed} and again requires an additional local stopping rule. More importantly, the lower bound is only obtainable for fixed graphs, and there does not exist such a $\rho$ for time-varying graphs. Another idea is to use ``fast" communications between the time slots $k$ and $k+1$ to {\em distributedly} evaluate $d(k)$. Although this is possible with a finite number of communications, it certainly increases the communication overhead, and even offset the benefit of distributed algorithms.

Overall, the problem of designing effective rules to locally stop distributed algorithms in (\ref{basic_dis}) largely remains open, which is the focus of this work. 
\subsection{Local Stopping Scheme}
We are interested in designing local stopping rules of the following form
\begin{equation}\label{observer}
s_i(k+1) = \mathcal{Q}_i(x_{\mathcal{N}_i^-\cup\{i\}}(k),s_{\mathcal{N}_i^-\cup\{i\}}(k)),
\end{equation}
where $s_i(k)$ is embeded  in node $i$ to monitor iterations of distributed algorithms of (\ref{basic_dis}).  If a given performance of (\ref{basic_dis}) is achieved, i.e., there exists a performance dependent function $\cJ_i(\cdot)$ such that
\begin{equation}
\mathcal{J}_i(s_i(k))\geq 0,\label{rule}
\end{equation}
then node $i$ stops updating and will no longer broadcast any information to its out-neighbors. Clearly, $s_i(k)$ can be distributedly implemented as in (\ref{basic_dis}), and synchronize with the original distributed algorithm, see Fig. \ref{scheme} for an illustration of implementing a local stopping rule in node $i$.
\begin{figure}[htp]
\centering
\includegraphics[width=7.4cm]{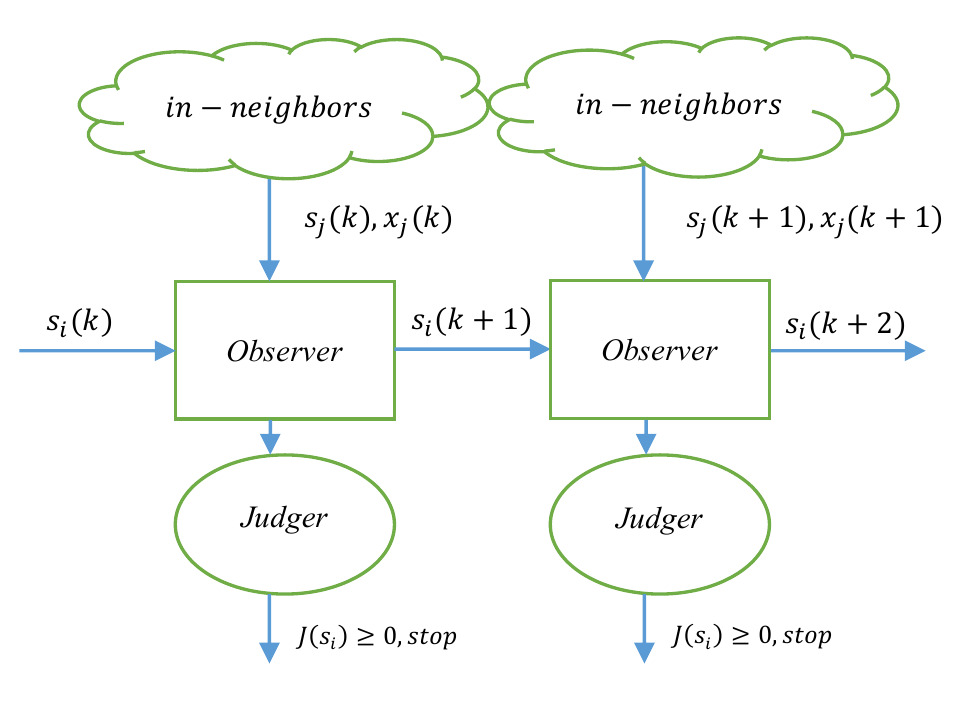}
\caption{Local stopping rule in node $i$.}
\label{scheme}
\end{figure}

Since $s_i(k)$ needs to be broadcast, the extra communication overhead is now decided by $s_i(k)$.  The number of bits to represent $s_i(k)$ should be as small as possible. While to ensure fast detection, it is highly desirable to easily compute both $\cQ_i(\cdot)$ and $\cJ_i(\cdot)$. Usually, the more information we obtain from the network topology, the simper form of these two functions are. Moreover,  the stopping rule of (\ref{observer}) and (\ref{rule}) needs to be sensitive, meaning that it is able to quickly detect the attainability of the performance quality of distributed algorithms.

In the sequel, we design a local stopping rule of the form (\ref{observer}) and (\ref{rule}) for the distributed consensus algorithm in (\ref{consensus_dis}).  The case of generic distributed algorithms will be studied in the full version of this work.

\section{Local Stopping Rules for Consensus Algorithms}
\label{consensus}
\subsection{Notions of $\epsilon$-consensus}
Obviously, the consensus algorithm in (\ref{consensus_dis}) can be written in the following matrix form
\begin{equation}
\label{systema}
x(k+1)=Ax(k),
\end{equation}
where $A=(\alpha_{ij})\in\mathbb{R}^{N\times N}$ is {\em row stochastic}. It follows from Gershgorin's disc theorem \cite{horn2012matrix} that $A^k\rightarrow1_Nv^T$ with $1_N^Tv=1$ if one is the simple eigenvalue of $A$, which implies that $x(k)\rightarrow 1_Nv^Tx(0)$, i.e., $|x_i(k)-x_j(k)|\rightarrow0$ and consensus is asymptotically achieved. By \cite{ren2005consensus},  one is the simple eigenvalue of $A$ if and only if $\cG$ associated with $A$ contains a spanning tree.  To examine the asymptotic behavior of (\ref{systema}), we introduce the concept of global $\epsilon$-consensus.

\begin{defi}[global $\epsilon$-consensus] Given $\epsilon>0$, the state vector $x=\left[x_1,\ldots,x_N\right]$ of the network is said to reach global $\epsilon$-consensus if $|x_i-x_j|<\epsilon$ for any $ i,j\in\cV$.
\end{defi}

The main objective of this paper is to provide a local stopping rule of the form (\ref{observer}) and (\ref{rule}) to ensure that the global $\epsilon$-consensus is achieved.  However, an individual node cannot directly check its attainability. Instead, by directly using states of its in-neighbors and without inducing extra communication, each node can easily check {\em local} $\epsilon$-consensus below.
 \begin{defi}[local $\epsilon$-consensus] Given $\epsilon>0$, if the state of node $i$ satisfies $max_{j\in\mathcal{N}_i^-}|x_i-x_j|<\epsilon$, node $i$ is said to achieve local  $\epsilon$-consensus.\end{defi}
A striking distinction between local $\epsilon$-consensus and global $\epsilon$-consensus lies in that the later one is attractive, meaning that once the network attains global $\epsilon$-consensus, it will stay in this state afterwards.  To elaborate it, let $x_u(k)=\max_i\{x_i(k)\}$, and $x_l(k)=\min_i\{x_i(k)\}$. It follows from (\ref{consensus_dis}) that $ x_u(k)\ge x_i(k+1)$ and  $x_j(k+1)
\ge x_l(k)$ for all $i,j\in\cV$. Then $ |x_i(k+1)-x_j(k+1)|\le  x_u(k)-x_l(k)=\max_{i,j}|x_i(k)-x_j(k)|$. That is, once global $\epsilon$-consensus is achieved at time slot $k$, it will continue to hold at next time slot.

 While  global $\epsilon$-consensus implies local $\epsilon$-consensus, the converse does not hold. To bridge the gap of the above two notions of $\epsilon$-consensus,  we further introduce {\em uniformly local $\epsilon$-consensus}.

\begin{defi}[uniformly local $\epsilon$-consensus]
Given $\epsilon>0$,  if all the nodes achieve local $\epsilon$-consensus, the network achieves uniform $\epsilon$-consensus.
\end{defi}

Again, global $\epsilon$-consensus implies uniformly local $\epsilon$-consensus. On the other side, we obtain the following result. 
\begin{lem}\label{equivalence} Suppose that $\cG$ is strongly connected. If the state of the network is uniformly local $\epsilon$-consensus, it is also in the state of global $(\epsilon D)$-consensus, where $D$ is the diameter of the network.
\end{lem}
\begin{proof} Given two arbitrary nodes $i,j\in\mathcal{V}$, then there exists a directed path with edges $(i,i_1)\in\cE, (i_1,i_2)\in\cE,\ldots, (i_m,j)\in\cE$ and $m+1\leq D$. For convenience, denote $i_0=i$ and $i_{m+1}=j$. Since the network is in uniformly local $\epsilon$-consensus, then $|x_{i_k}-x_{i_{k+1}}|<\epsilon$ holds for $k=0,\ldots,m$. Therefore, $|x_i-x_j|=|\sum_{k=0}^m(x_{i_k}-x_{i_{k+1}})|\leq\sum_{k=0}^m|x_{i_k}-x_{x_{k+1}}|<(m+1)\epsilon\leq\epsilon D$, which indicates that the network is in  $(\epsilon D)$-consensus.
\end{proof}

If $\cG$ is strongly connected, the above shows that uniformly local $\epsilon$-consensus is essentially equivalent to global $\epsilon$-consensus. 
Then, the remaining problem is how to locally check the attainability of uniformly local $\epsilon$-consensus. One may consider a straightforward idea by using a certain number of consecutively attaining local $\epsilon$-consensus. However, this number cannot be settled as a prior and in fact depends on the initial network state, which is illustrated in the following example.

\begin{exmp} Consider a strongly connected ring network with $3$ nodes. The associated weighting matrix $A$ and initial network state are given by
\begin{equation}
A=\left[\begin{matrix}0.5&0.5&0\\
     0&0.999&0.001\\
     0.5&0&0.5
\end{matrix}\right], ~\text{and}~ x(0)=\left[\begin{matrix}0\\0\\100\end{matrix}\right].\nonumber
\end{equation}

Once can readily verify that $x(3)=[0.1, 0.1748, 12.525]^T$ and node $1$ attains local $0.5$-consensus at any time slot. However,  either uniformly local $0.5$-consensus or global $0.5$-consensus is not achieved before time slot $k\leq 3$.  It is easy to observe that the attainability of uniformly local $\epsilon$-consensus depends on the initial state vector, and cannot be simply checked via a fixed number of consecutively attaining local $\epsilon$-consensus.  

\end{exmp}

In the sequel, we shall make the following assumpton. 
\begin{ass} \label{connect} The graph $\cG$ is strongly connected.
\end{ass}

In fact, Assumption \ref{connect} is necessary. If $\cG$ is not strongly connected, there exists a node $i$ that is not reachable from node $j$. Then, node $i$ is unable to receive information from node $j$, which renders node $i$ generically unable to locally decide the stopping time. 

\subsection{Design of Local Stopping Methods}
To resolve the above mentioned issue for uniformly local $\epsilon$-consensus, our idea is to adopt a minimum-consensus algorithm in the form (\ref{observer}) so that  each node can track the minimum number of consecutive attaining local $\epsilon$-consensus. 

To this purpose, each node $i$ uses $y_i^{\epsilon}(k)$ to record the {\em latest} number of  consecutively attaining local $\epsilon$-consensus in the time slot $k$, i.e.,  
\beq
&&\hspace{-1cm}y_i^{\epsilon}(0)=0,\\
&&\hspace{-1cm}y_i^{\epsilon}(k+1) = 
\begin{cases}
y_i^{\epsilon}(k)+1,&\max_{j\in\mathcal{N}_i^-}|x_i(k)-x_j(k)|<\epsilon,\\
0,&\text{otherwise.}
\end{cases}\notag
\label{af}
\enq

Clearly, if $y_i^{\epsilon}(k)$ is positive,  node $i$ attains local $\epsilon$-consensus at time slot $k-1$.  If node $i$ is further able to locally check the positiveness of $\min_{i\in\cV}\{y_i^{\epsilon}(k)\}$, the problem is solved as $\min_{i\in\cV}\{y_i^{\epsilon}(k)\}>0$ implies that uniformly local $\epsilon$-consensus is achieved! To this end, each node may use ``fast" communications between time slot $k$ and $k+1$ to send one-bit data for $D$ times to its out-neighbors, and adopts the minimum-consensus algorithm, i.e.,
\bee
s_i(k,0)=\text{sign}(y_i^{\epsilon}(k)), ~\text{and}~s_i(k,m+1)=\min_{j\in\mathcal{N}_i^-\cup\{i\}}s_j(k,m),\label{min-consensu}
\ene 
where $\text{sign}(x)=1$ if $x>0$, and zero, otherwise.  

Specifically, node $i$ only needs to broadcast one-bit data to its out-neighbors. If $\cG$ is strongly connected, it is clear that  $s_i(k,D)=\min_{i\in\cV}\{s_i(k,0)\}$. That is, the minimum-consensus is achieved after $D$ rounds of communications. If $s_i(k,D)>0$, then node $i$ stops updating (\ref{consensus_dis})
at time slot $k$, and reports that the network attains uniformly local $\epsilon$-consensus. However, this approach is not applicable to time-varying graphs. For example, if $\cG(k)$ is not strongly connected, the minimum consensus is not achievable by using (\ref{min-consensu}).

Inspired also motivated by its limitation, we introduce another variable $z_i^\epsilon$  in node $i$, which plays a similar role of $s_i(k,m)$ of (\ref{min-consensu}).  Let $z_i^\epsilon(0)=0$, and  
\bee\label{auxz}
z_i^\epsilon(k+1) = \min_{j\in\mathcal{N}_i^-\cup\{i\}}{\{z_j^\epsilon(k),y_j^\epsilon(k)\}}+1.
\ene

Clearly, both the auxiliary variables $y_i^\epsilon$ in (\ref{af}) and $z_i^\epsilon$ in (\ref{auxz}) can be computed in a distributed way. 

Next, we show how to use $z_i^\epsilon$ to check the attainability of uniformly local $\epsilon$-consensus. The following result is straightforward, and its proof is omitted. 

\begin{lem}\label{iter_prop} Consider the distributed algorithms given by (\ref{af}) and (\ref{auxz}), the following statements hold.
\begin{enumerate}\renewcommand{\labelenumi}{\rm(\alph{enumi})}
\item For any $k_1\geq k_2\geq0$, then $y_i^\epsilon(k_1)\leq y_i^\epsilon(k_2)+k_1-k_2$ and $z_i^\epsilon(k_1)\leq z_i^\epsilon(k_2)+k_1-k_2$.
\item If $(i,j)\in\mathcal{E}$, then $y_j^\epsilon(k)\geq z_i^\epsilon(k+1)-1$. 
\item The network is uniformly local $\epsilon$-consensus at time slot $k$ if and only if $y_i^\epsilon(k+1)\geq1$ for all $i\in\mathcal{V}$.
\end{enumerate}
\end{lem}
\begin{thm}\label{thm1}  Consider the distributed algorithms given by (\ref{af}) and (\ref{auxz}) under Assumption \ref{connect}. If there exists $i\in\mathcal{V}$ such that $z_i^\epsilon (k) \ge  D+1$,   the network attains global $(\epsilon D)$-consensus at time slot $k$ from any initial network state. 
\end{thm}
\begin{proof} For any $j\neq i\in\mathcal{V}$,  it follows from  Assumption \ref{connect} that there exists a directed path $(i,i_1),(i_1,i_2),\ldots,(i_m,j)$ from node $j$ to $i$ with $m\leq D-1$. Since $z_i(k)\geq D+1$ and $i_1\in\mathcal{N}_i^-$, it follows from Lemma \ref{iter_prop}(a) and (b) that $z_{i_1}^\epsilon (k-1)\geq D$, and $y_{i_1}^\epsilon(k-1)\geq D$. Similarly, it follows that $z_{i_2}^{\epsilon}(k-2)\geq D-1$ and $y_{i_2}^\epsilon(k-2)\geq D-1$. Repeat the same steps, we obtain that $y_j^\epsilon(k-m-1)\geq D-m$. Note that $m+1\leq D$, it follows from Lemma \ref{iter_prop}(a) that $y_j^\epsilon(k- D)\geq y_j^\epsilon(k-m-1)+m+1- D\ge1$. As $j$ is arbitrary, it follows from Lemma \ref{iter_prop}(c) that the system is in uniformly local $\epsilon$-consensus at time $k- D-1$. By Lemma \ref{equivalence}, the system has already attained global $(\epsilon D)$-consensus at time slot $k- D-1$. Since the state of global $(\epsilon D)$-consensus is attractive, it is obvious that the network attains  global $(\epsilon D)$-consensus at time slot $k$.
\end{proof}

If the diameter $D$ of the network is unknown, we may replace it by using the network size $N-1$ and obtain the following result.

\begin{cor} Consider the distributed algorithms given by (\ref{af}) and (\ref{auxz}) under Assumption \ref{connect}. If there exists $i\in\mathcal{V}$ such that $z_i^\epsilon (k) \ge  N$,  then the network attains global $(\epsilon (N-1))$-consensus at time slot $k$ from any initial network state. 
\end{cor}

Clearly, it is almost impossible for all nodes in the network  to simultaneously detect the attainability of global $\epsilon$-consensus.  In practice, once the node $i$ locally determines the attainability of global $\epsilon$-consensus, it stops updating. In addition,  this node will inform its out-neighbors to stop updating and there is no longer need to broadcast its state. Thus, it maximally takes an additional $D$ time slots for other nodes to detect  the attainability of global $\epsilon$-consensus. 

Overall, the local stopping method for node $i$ is detailed in Algorithm \ref{alg1}.

\begin{algorithm}
  \caption{Local stopping method for uniformly local $\epsilon$-consensus}
  \label{alg1}
  \begin{algorithmic}\renewcommand{\algorithmicrequire}{ \textbf{Initialization:}}
 \REQUIRE
  For each node $i\in\mathcal{V}$, let $y_i^{\epsilon} \gets 0$, $z_i^{\epsilon}\gets 0$.
   \WHILE {true}
  \renewcommand{\algorithmicensure}{\quad\textbf{\quad Local information exchange:}}
  \ENSURE Every node $i$ broadcasts its state $x_i$, and $\min\{y_i^{\epsilon}, z_i^{\epsilon}\}$ to its out-neighbors.
  \renewcommand{\algorithmicensure}{\quad\textbf{\quad Local variables update:}}
  \ENSURE Each node $i$ receives $x_j$ and $\min\{y_j^{\epsilon}, z_j^{\epsilon}\}$ from its in-neighbors, i.e., $j\in\cN_i^{-}$, and let
  \STATE $z_i^{\epsilon}\gets \min_{j\in\mathcal{N}_i^-\cup\{i\}}\{y_j^{\epsilon}, z_j^{\epsilon}\}+1$.
  \IF {$z^{\epsilon}_i\geq D+1$}
  \STATE node $i$ reports uniformly local $\epsilon$-consensus and stops updating.
  \STATE \textbf{break}
  \ENDIF
 \IF{$\max_{j\in\mathcal{N}_i^-}|x_i-x_j|<\epsilon$}
  \STATE $y_i^{\epsilon}\gets y_i^{\epsilon}+1.$
  \ELSE
  \STATE $y_i^{\epsilon}\gets 0.$
  \ENDIF
  \ENDWHILE
  \end{algorithmic}
\end{algorithm}
\begin{rem} Given any $\delta>0$, let $\epsilon=\delta/D$ in Algorithm \ref{alg1}. Then, we are able to obtain global $\delta$-consensus by using the stopping method in Algorithm \ref{alg1} .  
\end{rem}
\section{Sensitivity Analysis}
\label{simulation}
Theorem \ref{thm1} confirms feasibility of Algorithm \ref{alg1} for locally detecting both uniformly local and global consensus. If the network attains uniformly local or global $\epsilon$-consensus, a node usually is unable to locally detect it immediately due to the lack of the whole network state. However, the detection task can be completed by every node under Algorithm \ref{alg1} after a finite number of extra time slots, which is defined as {\em response time} in this work. Clearly, a short response time is desirable in applications, and is zero if the node is able to monitor the whole network state simultaneously.  In this section, we are interested in evaluating the response time of Algorithm \ref{alg1}.  To this purpose, we use adopt the concept of {\em ergodic coefficient}. 

\begin{defi}[Ergodic coefficient] Given a nonnegative matrix $A=[a_{ij}]\in\mathbb{R}^{N\times N}$, its ergodicity coefficient is defined as 
\begin{equation}
\tau(A)=\min_{i\neq j}\sum_{k=1}^N\min\{a_{ik},a_{jk}\}.
\end{equation}
\end{defi}
For a row-stochastic matrix $A$, it clearly hods that $0\le \tau(A)\le 1$, and ergodic coefficient is important to quantify the convergence rate of consensus algorithm in (\ref{systema}). Since $x(k+j)=A^j x(k)$ for any $k, j\in\bN$, it follows from \cite{dong2016flocking} that 
\begin{equation*}
d(k+j)\leq(1-\tau(A^j))d(k),
\end{equation*}
where $d(k)=\max_{i,j}\{|x_i(k)-x_j(k)|\}$. 
By \cite{dong2016flocking}, we obtain the following lemma. 

\begin{lem}\label{coeff1} Under Assumption \ref{connect}, there exists a minimum $h\in\bN$ such that $0<\tau(A^h)<1$
and $h\le D$. 
\end{lem}
\begin{thm}\label{resbound} Under Assumption \ref{connect}, consider the local stopping method in Algorithm \ref{alg1}. Given any $\delta>0$, the response time $T_r$ for achieving global $\delta$-consensus  is upper bounded by
\begin{equation}\label{response}
T_r\leq h\lceil\frac{-\log~D}{\log~(1-\tau(A^h))}\rceil+D+1,
\end{equation}
where $h$ is given in Lemma \ref{coeff1}.
\end{thm}
\begin{proof} Under Assumption \ref{connect}, it is obvious that the network state dictated in (\ref{systema}) will asymptotically reach consensus. Then,  there exists a finite time slot $k_0>0$ such that the network will firstly reach global $\delta$-consensus. Clearly, $k_0$ depends on both the network structure in the form of $A$ and the initial network state. Thus, it is generically unavailable to an individual node. 

If the network state attains the uniformly local $(\delta/D)$-consensus, each node is able to locally detect it before time slot $D+1$ by using Algorithm \ref{alg1}.  In view of Lemma \ref{equivalence}, the network certainly reaches global $\delta$-consensus. 

Since $x(k_0+lh)=A^hx(k_0+(l-1)h)$, it follows from Lemma \ref{coeff1}  that 
\begin{eqnarray}
d(k_0+lh)&\leq&(1-\tau(A^h))d(k_0+(l-1)h)\nonumber\\
&\leq&(1-\tau(A^h))^l \delta.\nonumber
\end{eqnarray}

Let $l$ be selected such that $(1-\tau(A^h))^l \delta\le \delta/D$, i.e, $$l= \lceil\frac{-\log~D}{\log~(1-\tau(A^h))}\rceil,$$ the network reaches uniformly local $(\delta/D)$-consensus at time slot time $k_0+l h$. 

Finally, we obtain that $T_r\le lh +D+1$, which is explicitly given by
\begin{equation*}
T_r\leq h\lceil\frac{-\log~D}{\log~(1-\tau(A^h))}\rceil+D+1.
\end{equation*}
That is, the proof is completed. 
\end{proof}

\begin{rem} In (\ref{response}),  the first quantity in the right hand side (RHS) is applied to overcome the conservativeness between global $\delta$-consensus and uniformly local $\delta/D$-consensus.  The additional $D+1$ time slots are for an individual node to detect the attainability of uniformly local $\delta/D$-consensus. 
\end{rem}

The following example illustrates the sensitivity of the stopping method via the response time.

\begin{exmp}\label{exmp1} Consider the consensus algorithm (\ref{systema}) with
\begin{equation*}
A=\left[\begin{matrix}0.933&0.067&0&0\\
     0&0.722&0.129&0.149\\
     0&0&0.633&0.367\\
     0.111&0&0&0.889
\end{matrix}\right],~\text{and}~x(0)=\left[\begin{matrix}10\\7\\4\\0\end{matrix}\right].
\end{equation*}
\end{exmp}
In this example, the goal is to achieve global $0.5$-consensus. The simulation result in Fig. \ref{fig:exmp1} shows that the first time when stopping condition satisfies in some node is $k=24$, while the global $\epsilon$-consensus actually arrives at $k=18$. The result clearly illustrates the effectiveness of Algorithm \ref{alg1} since each node stops, locally. 
\begin{figure}[htbp]
  \centering
  \includegraphics[width=0.51\textwidth]{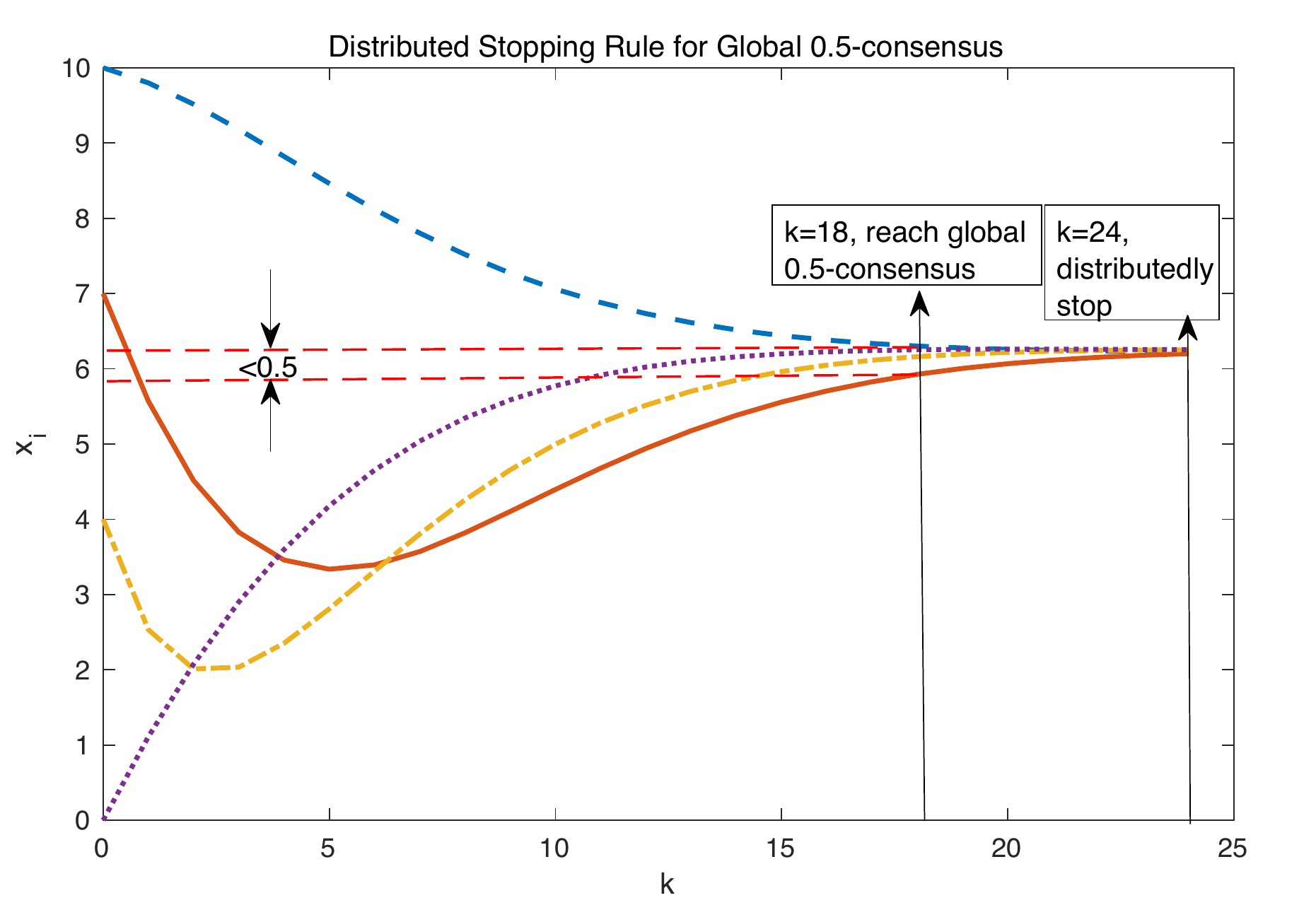}
  \caption{Simulation result of Example 1 as $\epsilon=0.5$.}\label{fig:exmp1}
\end{figure}

The response time for different levels of global consensus are listed in Table.1. Via calculation, we obtain $h=2$, $D=3$, and $\tau(A^h)=0.0594$. It follows from Theorem \ref{resbound} that $T_r\le 40$. As shown in Table.1, the largest response time is $10$, which is smaller than $40$ and verifies the correctness of Theorem \ref{resbound}.

\begin{table}[t!]
\label{table1}
\centering
\textbf{Table 1:} Response time for different levels of global consensus
\begin{tabular}{cccc}
\hline
Level of consensus  &Consensus time&Stopping time& Response time \\
\hline
1&15&23&8\\
0.5&18&24&6\\
0.1&23&30&7\\
0.01&34&44&10\\
0.001&44&51&7\\
\hline
\end{tabular}
\end{table} 

\section{Conclusion}
\label{conclusion}
In this paper, we formally defined the local stopping problem   for distributed algorithms. Focusing on the consensus algorithm, a distributed stopping method was designed for each nodes to locally check consensus, and its sensitivity is quantified. The method can be extended to time-varying network and is robust to communication delays, which will be presented in the journal version.

\bibliographystyle{IEEEtrans}
\bibliography{mybibf}

\begin{thebibliography}{10}
\providecommand{\url}[1]{#1}
\csname url@samestyle\endcsname
\providecommand{\newblock}{\relax}
\providecommand{\bibinfo}[2]{#2}
\providecommand{\BIBentrySTDinterwordspacing}{\spaceskip=0pt\relax}
\providecommand{\BIBentryALTinterwordstretchfactor}{4}
\providecommand{\BIBentryALTinterwordspacing}{\spaceskip=\fontdimen2\font plus
\BIBentryALTinterwordstretchfactor\fontdimen3\font minus
  \fontdimen4\font\relax}
\providecommand{\BIBforeignlanguage}[2]{{%
\expandafter\ifx\csname l@#1\endcsname\relax
\typeout{** WARNING: IEEEtran.bst: No hyphenation pattern has been}%
\typeout{** loaded for the language `#1'. Using the pattern for}%
\typeout{** the default language instead.}%
\else
\language=\csname l@#1\endcsname
\fi
#2}}
\providecommand{\BIBdecl}{\relax}
\BIBdecl

\bibitem{sheu2008distributed}
J.-P. Sheu, P.-C. Chen, and C.-S. Hsu, ``A distributed localization scheme for
  wireless sensor networks with improved grid-scan and vector-based
  refinement,'' \emph{IEEE transactions on mobile computing}, vol.~7, no.~9,
  pp. 1110--1123, 2008.

\bibitem{bazerque2010distributed}
J.~A. Bazerque and G.~B. Giannakis, ``Distributed spectrum sensing for
  cognitive radio networks by exploiting sparsity,'' \emph{IEEE Transactions on
  Signal Processing}, vol.~58, no.~3, pp. 1847--1862, 2010.

\bibitem{richards2002coordination}
A.~Richards, J.~Bellingham, M.~Tillerson, and J.~How, ``Coordination and
  control of multiple uavs,'' in \emph{AIAA guidance, navigation, and control
  conference, Monterey, CA}, 2002.

\bibitem{dong2016flocking}
J.-G. Dong and L.~Qiu, ``Flocking of the cucker-smale model on general
  digraphs,'' \emph{IEEE Transactions on Automatic Control}, 2016.

\bibitem{acemoglu2011opinion}
D.~Acemoglu and A.~Ozdaglar, ``Opinion dynamics and learning in social
  networks,'' \emph{Dynamic Games and Applications}, vol.~1, no.~1, pp. 3--49,
  2011.

\bibitem{yadav2007distributed}
V.~Yadav and M.~V. Salapaka, ``Distributed protocol for determining when
  averaging consensus is reached,'' in \emph{45th Annual Allerton Conference},
  2007, pp. 715--720.

\bibitem{manitara2014distributed}
N.~E. Manitara and C.~N. Hadjicostis, ``Distributed stopping for average
  consensus in directed graphs via a randomized event-triggered strategy,'' in
  \emph{6th International Symposium on Communications, Control and Signal
  Processing}.\hskip 1em plus 0.5em minus 0.4em\relax IEEE, 2014, pp. 483--486.

\bibitem{cortes2008distributed}
J.~Cort{\'e}s, ``Distributed algorithms for reaching consensus on general
  functions,'' \emph{Automatica}, vol.~44, no.~3, pp. 726--737, 2008.

\bibitem{olfati2004consensus}
R.~Olfati-Saber and R.~M. Murray, ``Consensus problems in networks of agents
  with switching topology and time-delays,'' \emph{IEEE Transactions on
  automatic control}, vol.~49, no.~9, pp. 1520--1533, 2004.

\bibitem{ren2005consensus}
W.~Ren and R.~W. Beard, ``Consensus seeking in multiagent systems under
  dynamically changing interaction topologies,'' \emph{IEEE Transactions on
  automatic control}, vol.~50, no.~5, pp. 655--661, 2005.

\bibitem{parsegov2016novel}
S.~E. Parsegov, A.~V. Proskurnikov, R.~Tempo, and N.~E. Friedkin, ``Novel
  multidimensional models of opinion dynamics in social networks,'' \emph{IEEE
  Transactions on Automatic Control}, 2016.

\bibitem{nedic2009distributed}
A.~Nedi{\'c} and A.~Ozdaglar, ``Distributed subgradient methods for multi-agent
  optimization,'' \emph{IEEE Transactions on Automatic Control}, vol.~54,
  no.~1, pp. 48--61, 2009.

\bibitem{di2014distributed}
P.~Di~Lorenzo and S.~Barbarossa, ``Distributed estimation and control of
  algebraic connectivity over random graphs,'' \emph{IEEE Transactions on
  Signal Processing}, vol.~62, no.~21, pp. 5615--5628, 2014.

\bibitem{horn2012matrix}
R.~A. Horn and C.~R. Johnson, \emph{Matrix analysis}.\hskip 1em plus 0.5em
  minus 0.4em\relax Cambridge University Press, 2012.

\end{thebibliography}
\end{document}